\documentclass[11pt]{article}

\usepackage[utf8]{inputenc}
\usepackage[T1]{fontenc}
\usepackage[english]{babel}

\usepackage{amsmath,amssymb,amsfonts,amsthm}
\usepackage{mathabx}

\newtheorem{lemma}{Lemma}
\newtheorem{theorem}{Theorem}

\theoremstyle{definition}
\newtheorem{definition}{Definition}
\theoremstyle{proof}

\usepackage{algorithm}
\usepackage[noend]{algpseudocode}

\usepackage{cleveref}

\crefname{lemma}{Lemma}{Lemmas}
\Crefname{lemma}{Lemma}{Lemmas}
\crefname{theorem}{Theorem}{Theorems}
\Crefname{theorem}{Theorem}{Theorems}
\crefname{observation}{Observation}{Observations}
\Crefname{observation}{Observation}{Observations}
\crefname{definition}{Definition}{Definitions}
\Crefname{definition}{Definition}{Definitions}
\crefname{section}{Section}{Sections}
\Crefname{section}{Section}{Sections}
\crefname{figure}{Figure}{Figures}
\Crefname{figure}{Figure}{Figures}
\crefname{algorithm}{Algorithm}{Algorithms}
\Crefname{algorithm}{Algorithm}{Algorithms}
\crefname{equation}{}{}
\Crefname{equation}{}{}

\renewcommand{\subset}{\subseteq}

\newcommand{\ceil}[1]{\left\lceil{#1}\right\rceil}
\newcommand{\floor}[1]{\left\lfloor{#1}\right\rfloor}

\newcommand{\eps}{\varepsilon}

\newcommand{\abs}[1]{\left | #1 \right |}
\newcommand{\set}[1]{\left \{ #1 \right \}}

\newcommand\drop[1]{}

\newcommand\dfs{\operatorname{dfs}}

\usepackage[pdftex]{graphicx}
\usepackage[usenames,dvipsnames,table]{xcolor}
\definecolor{shade}{RGB}{235,235,235}

\usepackage{todonotes}

\newcommand{\mbtk}[1]{#1}

\newcommand*\samethanks[1][\value{footnote}]{\footnotemark[#1]}
\usepackage{authblk}
\usepackage{fullpage}

\title{A simple and optimal ancestry labeling scheme for trees}
\author{Søren Dahlgaard\thanks{Research partly supported by Thorup's
    Advanced Grant from the Danish Council for Independent Research
    under the Sapere Aude research career programme.}}
\author{Mathias Bæk Tejs Knudsen\samethanks\ \thanks{Research
    partly supported by the FNU project AlgoDisc -
    Discrete Mathematics, Algorithms, and Data
    Structures.}}
\author{Noy Rotbart }

\affil{
 Department of Computer Science, University of Copenhagen\\
  Universitetsparken 5, 2100 Copenhagen\\
  \texttt{\{soerend,knudsen,noyro\}@di.ku.dk} }

\begin{document}
\maketitle
\begin{abstract}
We present a   $\lg n + 2 \lg \lg n+3$ ancestry labeling scheme for trees.
The problem was first presented by Kannan et al. [STOC 88']  along with  a simple $2 \lg n$ solution.
Motivated by applications to XML files, the label size was  improved incrementally over the course of more than 20 years by a series of papers.
The last, due to Fraigniaud and Korman [STOC 10'], presented an   asymptotically optimal $\lg n + 4 \lg \lg n+O(1)$ labeling scheme  using non-trivial tree-decomposition techniques.
By  providing a  framework generalizing interval based labeling schemes, we obtain a simple, yet asymptotically optimal solution to the problem.
Furthermore, our labeling scheme is attained by a small  modification of  the original $2 \lg n$ solution.
\end{abstract}



\section{Introduction}\label{Sec:Intro}
The concept of \emph{labeling schemes}, introduced by Kannan, Naor and Rudich \cite{Kannan92}, is a  method to  assign bit strings, or  labels,  to the vertices of a graph   such that a query  between vertices can be inferred directly from the assigned labels, without using a  centralized data structure.
A labeling scheme for a family of graphs $\mathcal{F}$  consists of an \emph{encoder} and a \emph{decoder}.
Given a graph $G \in \mathcal{F}$, the encoder  assigns labels to each node in $G$, and the  decoder  can infer  the query given only a set of labels.
The main quality measure for  a labeling scheme  is the size of the largest  label size it assigns to a node of any graph of the entire family.
One of the  most well studied questions in the context of labeling schemes is the
\emph{ancestry problem}.  An ancestry  labeling scheme for the family of rooted  trees of $n$ nodes $\mathcal{F}$ assigns  labels to a tree $T \in \mathcal{F}$ such that given the labels $\ell(u), \ell(v)$ of any two nodes $u,v\in T$, one can determine whether $u$ is an \emph{ancestor} of $v$ in $T$.



Improving the label size for this question is highly motivated by  XML search engines.
An XML document can be viewed as a tree and queries over such documents amount to testing ancestry relations between nodes of these trees
\cite{Abiteboul99,Deutsch99,cohen2010labeling}.
Search engines process queries using an index structure summarizing the ancestor relations.
 It is imperative to the performance of such engines, that as much as possible of this index can
reside in main memory. The big size of web data thus implies that even a small
reduction in label size may significantly improve the memory cost and
performance. A more detailed explanation can be found in \cite{Abiteboul06}.

One solution to this problem, dating back to at least '74 by
Tarjan~\cite{Tarjan74}  and used  in the seminal paper in '92 by Kannan  et al.~\cite{Kannan92} is the following:
Given an $n$ node tree $T$  rooted in $r$, perform a DFS traversal and  for each $u \in T$ let $\dfs(u)$ be the
index of $u$ in the traversal.
Assign the label of $u$ as  the pair $\ell(u) = (\dfs(u), \dfs(v))$,
where   $v$ is the descendant of $u$ of largest  index in the DFS traversal.
Given two labels $\ell(u) = (\dfs(u), \dfs(v))$ and $\ell(w)$,
$u$ is an ancestor of $w$ iff $\dfs(u)\le \dfs(w) \le \dfs(v)$.
In other words the  label $\ell(u)$ represents  an interval, and ancestry is determined by interval containment.
Since every DFS index is a number  in $\{1,\ldots, n\}$
The  size of a label  assigned by this  labeling scheme is  at most $2\lg n$.\footnote{Throughout
this paper we use $\lg n$ to denote the base $2$ logarithm $\lg_2 n$.}

This labeling scheme, denoted  \textproc{Classic} from hereon,
was the first in a long line of research to minimize the label size to $1.5 \lg n$~\cite{abiteboul2001compact}, $\lg n +O(\lg n/ \lg \lg n)$~\cite{Thorup01}\footnote{The paper discussed proves this bound on routing, which can be used to determine ancestry.},   $\lg n+O(\sqrt{\lg n})$~\cite{alstrup2002improved,Abiteboul06}
and finally $\lg n + 4\lg \lg n+O(1)$~\cite{Korman10}, essentially matching  a lower bound of $\lg n + \lg \lg n-O(1)$~\cite{Alstrup05}.
Additional results for this labeling scheme are a $\lg n+O(\lg \delta)$
labeling scheme for trees of depth at most $\delta$~\cite{Fraigniaud10} and investigation on a dynamic variant with some pre-knowledge on the structure of the constructed  tree~\cite{cohen2010labeling}.
The asymptotically optimal  labeling scheme~\cite{Korman10} also  implies the existence of a \emph{universal poset} of size $O(n^k\lg^{4k} n)$.

Labeling schemes for other functions were considered. Among which are  adjacency \cite{Kannan92,Alstrup02,Alstrup14}, routing \cite{Fraigniaud01,Thorup01}, nearest common ancestor \cite{Alstrup02NCA,alstrup2013near} connectivity~\cite{Korman10connectivity}, distance~\cite{gavoillea2004distance}, and flow~\cite{katz2004labeling}.

\subsection{Our contribution}
We present a \emph{simple} ancestry labeling scheme of size  $\lg n + 2\lg\lg n + 3$.
Similarly to~\cite{Abiteboul06,Korman10} our labeling  scheme is based on assigning an interval to
each node of the tree. Our labeling scheme can be seen as an extension of the  \textproc{Classic} scheme described above, with the key difference that rather than storing the exact size of the interval, we store only an approximation thereof.
In order to store only this approximation, our scheme assigns intervals larger than needed, forcing us to use a range larger than $1, \dots ,n$.
Our main technical contribution is  to minimize the label size by balancing the approximation ratio with the size of the range required to accommodate the resulting labels.
While it is a challenge to prove the mathematical properties of our labeling scheme, describing  it can be done in a few lines of pseudocode.

The simplicity of our labeling scheme  contrast the  labeling scheme
of~\cite{Korman10} , which relies among other things, on a highly nontrivial
tree decomposition that must precede it's encoding. As a concrete example, our
encoder can be implemented using a single DFS traversal, and
results in a label size of $\lg n + 2\lg\lg n
+ O(1)$ bits compared to $\lg n + 4\lg\lg n + O(1)$. However, for trees of
constant depth, our scheme has size $\lg n + \lg\lg n + O(1)$, which is worse
than the bound achieved in \cite{Fraigniaud10}.

The paper is structured as follows:
 In \cref{sec:intervals} we present a general framework for describing ancestry
 labeling schemes based on the notion of \emph{left-including intervals},
     which we introduce in \Cref{defn:valid}.
\Cref{classificationNecessary,classificationSufficient} show the correctness of all labeling schemes which construct labels by assigning
\emph{left-including intervals} to nodes.
We illustrate the usefulness of this framework by using it to describe \textproc{Classic} in ~\Cref{sec:classic}.
Finally, In \cref{sec:new}  we use the framework to construct our new  approximation-based ancestry  labeling scheme.


\subsection{Preliminaries}\label{sec:prelims}

\newcommand{\ao}{\overline{a}}
\newcommand{\bo}{\overline{b}}
We use the notation  $[n] = \set{0,1,\ldots,n-1}$ and denote the concatenation of two bit strings $a$ and $b$ by $a \circ b$.
Let $T =(V,E)$ be a tree rooted in $r$, where $u,v \in V$.
The node $u$ is an \emph{ancestor} of $v$ if $u$ lies on the unique path from the root to
$v$, and we call  $v$  a \emph{descendant} of $u$ iff $u$ is an ancestor of $v$.
We denote the subtree rooted in $u$ as  $T_u$,  i.e.~the tree consisting
of all descendants of $u$, and stress that a node is both an ancestor and descendant of itself.
The encoding and decoding running time are described under  the word-RAM model.

We denote the \emph{interval} assigned to a node $u$ by  $I(u) = [a(u), b(u)]$, where
$a(u)$ and $b(u)$ denote the lower and upper part of the interval,
respectively.
We also define $\ao(u)$ and $\bo(u)$ to be the maximum value of
$a(v)$ respectively $b(v)$, where $v$ is a descendant of $u$ (note that this
includes $u$ itself). We will use the following notion:
\begin{definition}\label{defn:valid}
    Let $T$ be a rooted tree and $I$ an interval assignment defined on $V(T)$.
    We say that the interval assignment $I$ is \emph{left-including} if for each $u,v\in
    T$ it holds that $u$ is an ancestor of $v$ iff $a(v)\in I(u)$.
\end{definition}
In contrast to  \cref{defn:valid}, the  literature surveyed~\cite{Kannan92,abiteboul2001compact,alstrup2002improved,Korman10}   considers  intervals where
$u$ is an ancestor of $v$ iff $I(v)\subseteq I(u)$, i.e. the interval of a descendant node is fully contained in the interval of the ancestor.
This distinction is amongst the unused leverage points which we will use to arrive at our new labeling scheme.


\section{A framework for interval based labeling schemes}
\label{sec:intervals}

In this section we introduce a framework for assigning intervals to tree nodes.
We will see in \cref{sec:classic,sec:new} how this framework can be used to
describe ancestry labeling schemes. The framework  relies heavily on the
values defined in \cref{sec:prelims}, namely $a(u), b(u), \ao(u), \bo(u)$.
An illustration of these values is found in  \cref{fig:example} below.
The interval $[\ao(u),\bo(u)]$  can be seen as a slack interval from which
$b(u)$ can be chosen. This  will prove useful  in \cref{sec:new}.

\begin{figure}[htbp]
    \centering
    \includegraphics[width=.8\textwidth]{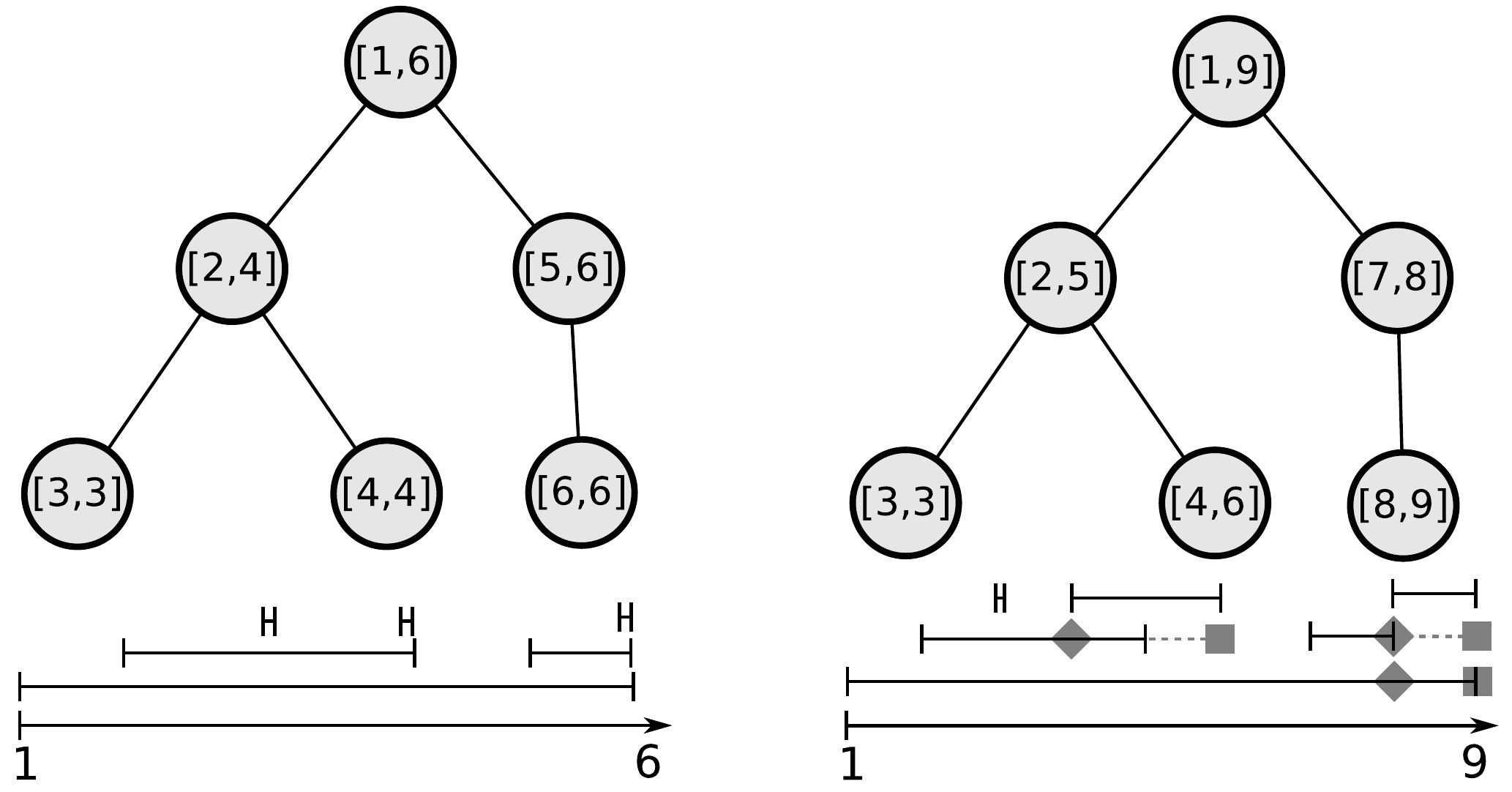}
    \caption{Two examples of left-including interval assignments to a tree.
    Left: a left-including assignment as used for \textproc{Classic} in the
    introduction corresponding to $b(u) = \ao(u)$.
    Right: a different left-including assignment for the same tree.
    For internal nodes where $\ao(u)$ and $\bo(u)$ do not coincide with $b(u)$,
we have marked these by a gray diamond and square respectively.}
    \label{fig:example}
\end{figure}

The following lemmas contain necessary and sufficient conditions for interval assignments satisfying the left inclusion property. 

\begin{lemma}
	\label{classificationNecessary}
    Let $T$ be a rooted tree and $I$ a left-including interval assignment defined on
    $V(T)$.
	Then the following is true:
	\begin{enumerate}
		\item For each $u \in T$, $b(u) \ge \ao(u)$.
		\item For each $u \in T$ and $v \in T_u \backslash \set{u}$ a descendant of $u$, $a(v) > a(u)$.
		\item For each $u \in T$, $[a(u), \bo(u)] = \bigcup_{v \in T_u} I(v) = \bigcup_{v \in T_u} [a(v), b(v)]$
		\item For any two distinct nodes $u,v \in T$ such that $u$ is not an ancestor of $v$ and $v$ is not
		      an ancestor of $u$ the intervals $[a(u), \bo(u)]$ and $[a(v), \bo(v)]$ are disjoint.
	\end{enumerate}
\end{lemma}
\begin{proof}
We prove that if the set of intervals is left-including each of the four statements must hold. We finish the proof
of each statement with a $\diamond$.

	1.
	Let $u \in T$. For any descendant $v \in T_u$ of $u$ we know that $a(v) \in [a(u), b(u)]$.
	In particular $a(v) \le b(u)$. Since this holds for any $v \in T_u$ we conclude that
	$\ao(v) = \max_{v \in T_u} \set{a(v)} \le b(u)$.
	$\hfill \diamond$

	2.
	Let $u \in T$ and $v \in T_u \backslash \set{u}$ be a descendant of $u$. Since $v$ is a descendant of $u$ we see that
	$a(v) \in [a(u), b(u)]$ and in particular $a(v) \ge a(u)$. Since $u$ is a not a descendant
	of $v$ we know that $a(u) \notin [a(v), b(v)]$ and in particular $a(u) \neq a(v)$. Hence
	$a(v) > a(u)$.
	$\hfill \diamond$

	3.
	Let $u \in T$. For any $v \in T_u$ $a(v) \in I(u)$ and hence $I(u)$ and $I(v)$ must have a
	non-empty intersection. Hence $\bigcup_{v \in T_u} I(v)$
	must be the union of overlapping intervals, and therefore it is an interval. Since $a(v) \ge a(u)$ for
	any $v \in T_u$ by part 2, the minimum value in the interval must be $a(u)$. The maximal value must be
	$\max_{v \in T_u} \set{b(v)}$ which is $\bo(u)$ by definition. Hence the interval must be $[a(u),\bo(u)]$.
	$\hfill \diamond$

	4.
	Let $u,v \in T$ be given as in the statement of the condition. Assume for the sake of contradiction that
	$[a(u), \bo(u)]$ and $[a(v), \bo(v)]$ are not disjoint. By part 3 this is equivalent to stating that
	the following sets are not disjoint:
	\[
		\bigcup_{w \in T_u} [a(w), b(w)],
		\bigcup_{z \in T_v} [a(z), b(z)]
	\]
	Hence there must exist $w \in T_u, z \in T_v$ such that $[a(w), b(w)]$ and $[a(z), b(z)]$ are
	overlapping. Assume wlog that $a(w) \le a(z)$. Since the intervals overlap $b(w) \ge a(z)$ and hence
	$a(z) \in [a(w), b(w)]$ and since the set of intervals is left-including this implies that $w$ is an ancestor
	of $z$. But this cannot be the case if none of $u,v$ is an ancestor of the other. Contradiction. Hence
	the assumption was wrong and the intervals are indeed disjoint.
	$\hfill \diamond$
\end{proof}

\begin{lemma}
	\label{classificationSufficient}
    Let $T$ be a rooted tree and $I$ an interval assignment defined on
    $V(T)$.
	If the following conditions are satisfied, then $I$ is a left-including interval
    assignment.
	\begin{enumerate}
		\item[  i] For each $u \in T$, $b(u) \ge \ao(u)$.
		\item[ ii] For each $u \in T$ and $v \in T$ a child of $u$, $a(v) > a(u)$.
		\item[iii] For any two siblings $u,v \in T$ the intervals $[a(u), \bo(u)]$ and $[a(v), \bo(v)]$ are disjoint.
	\end{enumerate}
\end{lemma}
\begin{proof}
	We prove that if the three conditions hold the set of intervals are left-including. First we prove that condition
	$ii$ implies condition $2$ from \Cref{classificationNecessary}:

	For each $u \in T$ and $v \in T_u \backslash \set{u}$ a descendant of $u$, $a(v) > a(u)$.

	\noindent
	Let $u \in T, v \in T_u \backslash \set{u}$ and consider the path from $u$ to $v$, say
	$u = v_0 \to v_1 \to \ldots \to v_k = v$ for positive integer $k$. Then $v_i$ is the child of $v_{i-1}$ for
	each $i=1,2,\ldots,k$. Therefore, by condition $ii$:
	\[
		a(u) = a(v_0) < a(v_1) < \ldots < a(v_k) = a(v)
	\]
	Which proves condition $2$.

	Fix distinct $u,v \in T$. We will
	show that $u$ is an ancestor of $v$ if and only if $a(v) \in I(u)$. We do this by considering
	three cases where $u$ is an ancestor of $v$, $v$ is an ancestor of $u$ and $v \neq u$, and none is
	an ancestor of the other respectively.

	\emph{Case 1: $u$ is an ancestor of $v$.} By condition $i$ $b(u) \ge \ao(u)$.
	Since $v \in T_u$ this implies $b(u) \ge a(v)$.
	By condition $2$ $a(v) > a(u)$. Therefore $a(u) < a(v) \le b(u)$ and $a(v) \in I(u)$.

	\emph{Case 2: $v$ is an ancestor of $u$ and $v \neq u$.} By condition $2$ $a(u) > a(v)$. Hence
	$a(v)$ is smaller than any value
	in the interval $[a(u), b(u)]$, and therefore $a(v) \notin I(u)$.

	\emph{Case 3: None of $u,v$ is an ancestor of the other.} By assumption $u,v$ must have a nearest
	common ancestor $w \in T$ which is neither $u$ or $v$. Let $u'$ be the child of $w$ that lies on the path
	from $w$ to $u$. Note that if $u$ is the child of $w$ then $u = u'$. Define $v'$ analogously. By condition $iii$
	$[a(u'), \bo(u')]$ and $[a(v'), \bo(v')]$ are disjoint. Hence either $\bo(u') < a(v')$ or
	$\bo(v') < a(u')$. Since both cases are handled similarly we will show how to handle the case $\bo(u') < a(v')$.
	By condition $2$ $a(v) > a(v')$ and hence $a(v) > \bo(u')$.
	Recall that $\bo(u')$ is the maximum of $b(z)$ where $z \in T_{u'}$. Since $T_u \subset T_{u'}$ this implies
	that $\bo(u) \le \bo(u')$. Hence $a(v) > \bo(u)$ as well and the intervals $[a(v),\bo(v)]$ and $[a(u), \bo(u)]$
	are disjoint. Hece $I(u)$ and $I(v)$ must be disjoint as well and $a(v) \notin I(u)$.
\end{proof}

\subsection{The framework}
We now consider a general approach for creating left-including interval
assignments.
For a node $u \in T$ and a positive integer $t$ we define the procedure
$\Call{Assign}{u, t}$
that assigns intervals to $T_u$ recursively and in particular, assigns
$a(u) = t$. For  pseudocode of  the procedure see~\cref{alg:general}.

\begin{algorithm}
	\caption{Assigning intervals to all nodes in the subtree $T_u$ rooted at
    $u$ ensuring $a(u) = t$.}
	\label{alg:general}
	\begin{algorithmic}[1]
		\Procedure{Assign}{$u, t$}
		\State $(a(u), \ao(u), b(u), \bo(u)) \gets (t,t,t,t)$ \label{alg:initAssign}
        \For{$v\in children(u)$}\label{alg:forloop}
			\State \Call{Assign}{$v, \bo(u)+1$}\label{alg:AssignNext}
			\State $ \left (\ao(u), \bo(u) \right ) \gets \left (\ao(v), \bo(v)\right )$ \label{alg:loopinvar}
		\EndFor
		\State Assign $b(u)$ such that $b(u) \ge \ao(u)$.\label{alg:assignb}
		\State $\bo(u) \gets \max \set{b(u), \bo(u)}$\label{alg:finishbo}
	\EndProcedure
	\end{algorithmic}
\end{algorithm}

\cref{alg:general} provides a general framework for assigning intervals
using a depth-first traversal.
We can use it to design an actual interval assignment by specifying: (1)
the way we choose $b(u)$, and (2) the order in which the children are
traversed. These specifications correspond to \cref{alg:assignb} and
\cref{alg:forloop}, respectively, and determine entirely the way the
intervals are assigned. It may seem counter-intuitive to pick $b(u) >
\bar{a}(u)$, but we will show that doing so in a systematic way, we are
able to describe the interval using fewer bits by limiting the choices for
$b(u)$. In the remainder of this paper, we will see how these
two decisions impact also the label size, and produce our claimed labeling
scheme.

We now show  that any ordering of the children  and any way of choosing $b(u)$ satisfying $b(u)\ge
\ao(u)$ generates a left-including interval assignment.

\begin{lemma}
	\label{correctOaOb}
	Let $T$ be a tree  rooted in  $r$.
    After running \Cref{alg:general} with  $\Call{Assign}{r, 0}$ the values of $\ao(u), \bo(u)$
	are correct, i.e. for all $u \in T$:
	\[
		\ao(u) = \max_{v \in T_u} \set{a(v)},
		\quad
		\bo(u) = \max_{v \in T_u} \set{b(v)}.
	\]
\end{lemma}
\begin{proof}
	We  prove the claim by induction on the size of the subtree $\abs{T_u}$,
    and prove that the invariant holds when the function call
    $\Call{assign}{u,t}$  is terminated.This is sufficient, since none of the involved variables are altered afterwards.

	As   base case $\abs{T_u} = 1$ and  $u$ is a leaf and the claim holds trivially.
	Assume next that the claim holds for  $\abs{T_u} \le m$ for some  integer $m \ge 1$,
	and let $u \in T$ be of size $\abs{T_u} = m+1$, and  $v_1,\ldots,v_k$ be the children of $u$ in the order they are processed on \Cref{alg:forloop}.
	First, note that the size of each $\abs{T_{v_1}} \dots \abs{T_{v_k}} \le m$.
	It follows that $\ao(v_i)$ and $\bo(v_i)$ are correct for the nodes $v_i$ where $1 \leq i \leq v_k$.
	Line 5 guarantees that $\bo(u)$ is at least as large as the largest $\bo(v_i)$, and that $\ao(u)$ is at least as large as the largest $\ao(v_i)$.
\end{proof}

The following Lemma is useful for showing several properties in the framework.
\begin{lemma}
	\label{BasicProperty}
Let $u$ be a node in a tree $T$ with children  $v_1 \dots v_k$. After running \Cref{alg:general} with parameters $\Call{Assign}{r,0}$ where $v_1 \dots v_k$ are processed in that order,  the following properties hold:
\begin{enumerate}
\item{$\bo(u)- a(u)+1 = \left ( \sum_{i=1}^k \bo(v_i) - a(v_{i}) + 1 \right )+1$.}
\item{$\ao(u)- a(u)+1  = \ao(v_k)- a(v_k)+ \left (\sum_{i=1}^{k-1} \bo(v_i) - a(v_i) +1 \right )+1$.}
\end{enumerate}

\end{lemma}
\begin{proof}
By the definition of \textproc{Assign} we see that for
	all $i = 1,\ldots,k-1$, $a(v_{i+1}) = 1 + \bo(v_i)$. Furthermore $a(v_1) = a(u) + 1$ and
	$\bo(v_k) = \bo(u)$. Hence:
	\begin{align*}
        \bo(u) - a(u) + 1 &=
		\bo(v_k) - a(v_1) + 2 \\ & =
		\left ( \sum_{i=2}^k \bo(v_i) - \bo(v_{i-1}) \right )
		+ \bo(v_1) - a(v_1) + 2
		\\
		& =
		\left ( \sum_{i=1}^k \bo(v_i) - a(v_{i}) + 1 \right )
		+ 1. \\
	\end{align*}
	The second equality follows by the same line of argument.
\qed
\end{proof}
\begin{theorem}
	\label{correctAlgo}
	Let $T$ be a  tree rooted in  $r$.
    After running \Cref{alg:general} with parameters $\Call{Assign}{r,0}$ the set of intervals produced are left-including.
\end{theorem}

\begin{proof}
Consider any node $u\in T$ and a call $\Call{Assign}{u,t}$. We will prove each of the
conditions of \cref{classificationSufficient}, which implies the theorem.
\begin{itemize}
    \item[  i] This condition is trivially satisfied by \cref{alg:assignb}.
    \item[iii] First, observe that any interval assigned to a node $w$ by a
        call to $\Call{Assign}{v,t}$ has $a(w)\ge t$, and by $i$ it has
        $\bo(w) \ge b(w) \ge a(w)$. Let $v_1,\ldots,v_k$ be the children in the order of
        the for loop in \cref{alg:forloop}. By
        \cref{alg:initAssign,alg:AssignNext,alg:loopinvar} we have $a(v_1) =
        a(u) + 1, a(v_2) = \bo(v_1) + 1, a(v_3) = \bo(v_2) + 1,\ldots, a(v_k) =
        \bo(v_{k-1}) + 1$, thus the condition is satisfied.
    \item[ ii] The first child $v$ of $u$ has $a(v) = t + 1 = a(u) + 1$. By the
        same line of argument as in $iii$ we see that all other children
        $w$ of $u$ must have $a(w) > a(v) = a(u) + 1$.\qed
\end{itemize}

\end{proof}

\section{The classic ancestry labeling scheme}
\label{sec:classic}
To get acquainted with the framework of \Cref{sec:intervals}, we use it to redefine \textproc{Classic}, the labeling scheme introduced in \Cref{Sec:Intro}.

Let $T$ be a tree rooted in $r$.
We first modify the function \textproc{Assign} to create
\textproc{Assign-Classic} such that the intervals $I(u) = [a(u), b(u)]$
correspond to the intervals of the \textproc{Classic} algorithm described in
the introduction. To do this we set $b(u) = \ao(u)$ in \cref{alg:assignb} and
traverse the children in any order in \cref{alg:forloop}.
We note that there is a clear distinction between an algorithm such as
\textproc{Assign-Classic} and an encoder. This distinction will be more clear
in \cref{sec:new}. We will need the following lemma to describe the encoder.

\begin{lemma}
	\label{classicInvariant}
    After $\Call{Assign-Classic}{u,t}$ is called the following invariant is
    true:
	\[
		\bo(u) - a(u) + 1 = \abs{T_u}
	\]
\end{lemma}
\begin{proof}
	We prove the claim by induction on $\abs{T_u}$.
	When $\abs{T_u} = 1$ $u$ is a leaf and hence $\bo(u) = a(u) = t$
	and the claim holds.

	Let $\abs{T_u} = m > 1$ and assume that the claim holds for all nodes with subtree size $< m$. Let
    $v_1,\ldots,v_k$ be the children of $u$. By \cref{BasicProperty} and the
    induction hypothesis we have:
    \begin{align*}
        \bo(u) - a(u) + 1
		& =
		\left ( \sum_{i=1}^k \bo(v_i) - a(v_{i}) + 1 \right )
		+ 1 \\
		& =
	    \left ( \sum_{i=1}^k \abs{T_{v_i}} \right )
		+ 1 =
		\abs{T_u}.
	\end{align*}
	This completes the induction. \qed
\end{proof}

\textbf{Description of the encoder:}
Let $T$ be an $n$-node tree rooted in $r$. We first invoke a call to
$\Call{Assign-Classic}{r,0}$. By \Cref{classicInvariant} we have $\bo(r) - a(r)
+ 1 = n$ and this implies $0\le a(u),b(u)\le n-1$ for every $u\in T$.
Let $x_u$
and $y_u$ be the encoding of $a(u)$ and $b(u)$ using exactly\footnote{This can be accomplished by padding with zeros if necessary.}
$\ceil{\lg n}$ bits respectively. We set the label of $u$ to be the
concatenation of the two bitstrings, i.e.~$\ell(u) = x_u\circ y_u$.

\textbf{Description of the decoder:} Let $\ell(u)$ and $\ell(v)$ be the labels of the nodes $u$ and $v$ in a tree $T$.
 By the definition of the encoder, the labels have the same size and it is $2z$
 for some integer $z \ge 1$. Let
$\ell(u) = x_u \circ y_u$ where $x_u$ and $y_u$ are the first and last $z$ bits of $\ell(u)$ respectively. Let
$a_u$ and $b_u$ the integers from $[2^z]$ corresponding to
the bit strings $x_u$ and $y_u$ respectively. We define $a_v$ and $b_v$ analogously. The decoder responds
\texttt{True}, i.e. that $u$ is the ancestor of $v$, iff $a_v \in [a_u, b_u]$.

The correctness of the labeling scheme follows from \Cref{correctAlgo} and the description of the decoder.

\section{An approximation-based approach}
\label{sec:new}

In this section we present the main result of this paper:
\begin{theorem}\label{thm:result}
 There exist  an ancestry labeling scheme of size $\ceil{\lg n} + 2\ceil{\lg \lg n} + 3$.
\end{theorem}

To prove this theorem, we   use the framework introduced in~\Cref{sec:intervals}.
The barrier in reducing the size of the \textproc{Classic} labeling scheme is
that the number of different intervals that can be assigned to a node is
$\Theta(n^2)$. It is impossible to encode so many different intervals without using at least
$2\lg n - O(1)$ bits. The challenge is therefore to find a smaller set of intervals $I(u) = [a(u), b(u)]$
to assign to the nodes. First, note that \cref{classificationNecessary}
points 2 and 4 imply that any two nodes $u,v$ must have $a(u)\ne a(v)$.
By considering the  $n$ node tree $T$ rooted in  $r$ where $r$  has $n-1$ children, we also see that there must be at least $n-1$ different
values of $b(u)$ (by \cref{classificationNecessary} point 4).
One might think that this implies the need for $\Omega(n^2)$ different intervals. This is, however, not the case.
We  consider a family of intervals, such that $a(u) = O(n)$ and the size of
each interval, $b(u)-a(u)+1$, comes from a much smaller set, $S$. Since
there are $O(n \abs{S})$ such intervals we are able to encode them using $\lg n+ \lg \abs{S} + O(1)$ bits.

We  now present a modification of \textproc{Assign} called
\textproc{Assign-New}. Calling $\Call{Assign-New}{r,0}$ on an $n$-node
tree $T$ with root $r$ will result in each $a(u)\in [2n]$ and $b(u)\in S$,
where $S$ is given by:
\begin{align}
	\label{eq:defS}
	S = \set{
			\floor{ \left ( 1 + \eps \right )^k }
		\mid
			k \in \left [ 4\ceil{\lg n}^2 \right ]
		}\ ,
\end{align}
where $\eps$ is the unique solution to the equation $\lg(1+\eps) = \left(
\ceil{\lg n} \right)^{-1}$. First, we  examine some properties of $S$:
\begin{lemma}
	\label{approxS}
	Let $S$ be defined as in \Cref{eq:defS}. For every $m \in \set{1,2,\ldots,2n}$ there exists $s \in S$ such that:
	\[
		m \le s < m(1+\eps)\ .
	\]
	Furthermore, $s = \floor{(1+\eps)^k}$ for some $k\in \left[ 4\ceil{\lg
	n}^2\right]$, and both $s$ and $k$ can be computed in $O(1)$ time.
\end{lemma}
\begin{proof}
	Fix $m \in \set{1,2,\ldots,2n}$.
	Let $k$ be the largest integer such that $\left ( 1 + \eps \right )^{k-1} <
    m$. Equivalently, $k$ is the largest integer such that:
	\[
		k-1 < \frac{\lg m}{\lg \left ( 1 + \eps \right )} =
		\left ( \lg m \right ) \cdot \ceil{\lg n}.
	\]
	In other words we choose $k$ as $\ceil{\left ( \lg m \right ) \cdot \ceil{\lg n}}$
	and note that $k$ is computed in $O(1)$ time.
	Since $\lg m \le \lg (2n) \le 2\lg n$: 
	\[
		k \le \ceil{2 (\lg n ) \cdot \ceil{\lg n}} \le 2 \ceil{\lg n}^2
		< 4 \ceil{\lg n}^2.
	\]
    By setting $s = \floor{(1+\eps)^k}$ we have $s \in S$. By the definition of $k$ we
    see that $(1+\eps)^k \ge m$ and thus also $s \ge m$.
    Similarly:
	\[
		m(1+\eps) >
		\left ( 1 + \eps \right )^{k-1} \cdot (1+\eps) =
		(1 + \eps)^k \ge s.
	\]
	This proves that $s \in S$ satisfies the desired requirement. Furthermore $s$ can be computed in $O(1)$
	time by noting that:
	\[
		s =
		\floor{(1+\eps)^k} =
		\floor{2^{\lg(1+\eps)k}} =
		\floor{2^{\ceil{\lg n}^{-1} \cdot k}}.
	\]
	\qed
\end{proof}

We now  define \textproc{Assign-New} by modifying
\textproc{Assign} in the following two ways.
First, we specify the order in which the children are traversed in
\cref{alg:forloop}. This is done in non-decreasing order of their subtree
size, i.e.~we iterate $v_1,\ldots,v_k$, where
$|T_{v_1}|\le\ldots\le|T_{v_k}|$.
Second, we choose $b(u)$ in \Cref{alg:assignb} as the smallest value, such that $b(u)\ge\ao(u)$ and $b(u)-a(u)+1 \in S$. This is
done by using \cref{approxS} with $m = \ao(u) - a(u) + 1$ and setting
$b(u) = a(u) + s - 1$. In order to do this we must have $m\le 2n$. To do this,
we show the following lemma corresponding to \cref{classicInvariant} in
\cref{sec:classic}.

\begin{lemma}
	\label{newInvariant}
    After $\Call{Assign-New}{u,t}$ is called the following invariants holds:
	\begin{align}
		\label{eq:aoClaim}
		\ao(u) - a(u) + 1 & \le \abs{T_u}(1 + \eps)^{\floor{\lg \abs{T_u}}}
		\\
		\label{eq:boClaim}
		\bo(u) - a(u) + 1 & \le \abs{T_u}(1 + \eps)^{\floor{\lg \abs{T_u}}+1}
	\end{align}
\end{lemma}
\begin{proof}
	We  prove the claim by induction on $\abs{T_u}$.
	When $\abs{T_u} = 1$, $u$ is a leaf, so $\bo(u) = \ao(u) = a(u) = t$
	and the claim holds.

	Now let $\abs{T_u} = m > 1$ and assume that the claim holds for all nodes with subtree size $< m$. Let
	$v_1,\ldots,v_k$ be the children of $u$ such that $\abs{T_{v_1}} \le \ldots \le \abs{T_{v_k}}$.
	First, we show that \Cref{eq:aoClaim} holds.
 By \cref{BasicProperty} we have the following expression for $\ao(u) - a(u) +
    1$:
	\begin{align}
		\label{eq:aoProof1}
		\ao(u) - a(u) + 1 =
		\left ( \ao(v_k) - a(v_k) + 1 \right ) +
		\left ( \sum_{i=1}^{k-1} \bo(v_i) - a(v_{i}) + 1 \right ) + 1.
	\end{align}
	It  follows from the induction hypothesis that:
	\begin{align}
		\label{eq:aoProof2}
		\ao(v_k) - a(v_k) + 1 \le
		\abs{T_{v_k}}(1 + \eps)^{\floor{\lg \abs{T_{v_k}}}}
		\le
		\abs{T_{v_k}}(1 + \eps)^{\floor{\lg \abs{T_u}}}.
	\end{align}
    Furthermore, by the ordering of the children, we have $\lg \abs{T_{v_i}}
    \le \lg \abs{T_u} - 1$ for every $i = 1,\ldots,k-1$. Hence:
	\begin{align}
		\label{eq:aoProof3}
		\bo(v_i) - a(v_i) + 1 \le
		\abs{T_{v_i}}(1 + \eps)^{\floor{\lg \abs{T_{v_i}}}+1}
		\le
		\abs{T_{v_k}}(1 + \eps)^{\floor{\lg \abs{T_u}}}.
	\end{align}
	Inserting \Cref{eq:aoProof2,eq:aoProof3} into \Cref{eq:aoProof1} proves invariant \Cref{eq:aoClaim}.

	Since $\bo(u) = \max\set{\bo(v_k), b(u)}$ we only need to upper bound $\bo(v_k) - a(u)+1$ and
	$b(u) - a(u)+1$. First we note that since $b(u)$ is chosen smallest possible such that $b(u) \ge \ao(u)$
	and $b(u)-a(u)+1 \in S$, it is guaranteed by \Cref{approxS} that:
	\[
		b(u) - a(u) + 1 <
		\left ( 1+ \eps \right )
		\left ( \ao(u) - a(u) + 1 \right ) \le
		\abs{T_u}(1 + \eps)^{\floor{\lg \abs{T_u}}+1}
	\]
	Hence we just need to upper bound $\bo(v_k) - a(u)+1$. First we note that just as in \Cref{eq:aoProof1}:
	\begin{align}
		\label{eq:boProof1}
		\bo(v_k) - a(u) + 1 =
		\left ( \sum_{i=1}^k \bo(v_i) - a(v_{i}) + 1 \right ) + 1
	\end{align}
	By the induction hypothesis, for every $i = 1,\ldots,k$:
	\begin{align}
		\label{eq:boProof2}
		\bo(v_i) - a(v_{i}) + 1
		\le
		\abs{T_{v_i}} (1+\eps)^{\floor{\lg \abs{T_{v_i}}}+1}
		\le
		\abs{T_{v_i}} (1+\eps)^{\floor{\lg \abs{T_u}}+1}
	\end{align}
	Inserting \Cref{eq:boProof2} into \Cref{eq:boProof1} gives the desired:
	\begin{align*}
		\bo(v_k) - a(u) + 1 \le
		1 +
		\sum_{i=1}^k \abs{T_{v_i}} (1+\eps)^{\floor{\lg \abs{T_u}}+1}
		\le
		\abs{T_u} (1+\eps)^{\floor{\lg \abs{T_u}}+1}.
	\end{align*}
	This completes the induction. \qed
\end{proof}

By \Cref{newInvariant} we see that for a tree $T$ with $n$ nodes  and $u \in T$:
\begin{align*}
	\ao(u) - a(u) + 1 \le
	\abs{T_u}(1 + \eps)^{\floor{\lg \abs{T_u}}} \le
	n \cdot 2^{\floor{\lg n} \lg(1+\eps)} \le
	n \cdot 2^1 = 2n.
\end{align*}
In particular, for any $u \in T$ we see that $a(u) \leq 2n$, and by \Cref{approxS} the function
\textproc{Assign-New} is well-defined.

We are now ready to describe the labeling scheme:

\textbf{Description of the encoder:}
Given an $n$-node tree $T$ rooted in $r$, the encoding algorithm works by
    first invoking a call to $\Call{Approx-New}{r,0}$. Recall that by
\cref{approxS} we find $b(u)$ such that $b(u) - a(u) + 1 =
\floor{(1+\eps)^k}$ as well as the value of $k$ in $O(1)$ time. For a node $u$, denote
the value of $k$ by $k(u)$
and let $x_u$ and $y_u$ be the bit strings representing $a(u)$ and $k(u)$
respectively, consisting of exactly
$\ceil{\lg(2n)}$ and $\ceil{\lg(4\ceil{\lg n}^2)}$ bits (padding with zeroes if
necessary). This is possible
since $a(u) \in [2n]$ and $k(u) \in \left [ 4\ceil{\lg n}^2 \right ]$.

For each node $u\in T$ we assign the label $\ell(u) = x_u\circ y_u$. Since
\[
	\ceil{\lg(2n)} = 1 + \ceil{\lg n}
	, \quad
	\ceil{\lg(4\ceil{\lg n}^2)} =
	2 + \ceil{2\lg(\ceil{\lg n})} =
	2 + \ceil{2\lg \lg n},
\]
the label size of this scheme is  $\ceil{\lg n} + \ceil{2 \lg \lg n} + 3$.

\mbtk{
\textbf{Description of the decoder:} Let $\ell(u)$ and $\ell(v)$ be the labels of the nodes $u$ and $v$ in a tree $T$.
By the definition of the encoder the labels have the same size and it is $s = z + \ceil{2 \lg z} + 3$
for some integer $z \ge 1$. By using that $s - \ceil{2 \lg s} - 3 = z - O(1)$ we can compute
$z$ in $O(1)$ time. We know that the number of nodes $n$ in $T$ satisfies $\ceil{\lg n} = z$. We can therefore
define $\eps$ to be the unique solution to $\lg(1 + \eps) = \ceil{\lg n}^{-1} = z^{-1}$. Let $x_u$ and $y_u$ be the first
$z+1$ bits and last $\ceil{2\lg z}+2$ bits of $\ell(u)$ respectively. We let $a_u$ and $k_u$ be the integers
in $[2^{z+1}]$ and $[4z^2]$ corresponding to the bit strings $x_u$ and $y_u$ respectively. We define $s_u$
as $\floor{(1+\eps)^{k_u}}$ and $b_u = s_u+a_u-1$. We define $a_v,b_v,k_v,s_v$ analogously. The decoder responds
\texttt{True}, i.e. that $u$ is the ancestor of $v$, iff $a_v \in [a_u, b_u]$.
}

\cref{thm:result} is now achieved by using the labeling scheme described
above. Correctness follows from \cref{correctAlgo}.


\bibliographystyle{amsplain}
\bibliography{bibl}

\end{document}